\documentclass{autart}
\usepackage{graphicx}      
\usepackage{natbib}        
\usepackage{amsmath}
\usepackage{mathptmx} 
\usepackage{times} 
\usepackage{amssymb}  
\usepackage{amsfonts}
\usepackage{mathtools}
\usepackage{color}
\usepackage{grffile}
\usepackage{subfigure}
\usepackage{epsfig} 
\usepackage{amsmath,lipsum} 
\usepackage[normalem]{ulem}
\usepackage{cuted}
\usepackage{enumerate}

\newtheorem{remark}{Remark}
\newtheorem{theorem}{Theorem}
\newtheorem{lemma}{Lemma}
\newtheorem{assumption}{Assumption}
\newtheorem{proposition}{Proposition}

\newtheorem{definition}{Definition}
\newenvironment{proof}{\par{\em{Proof: }}}{\hfill$\square$\\}

\newcommand{\matr}[1]{\begin{pmatrix} #1 \end{pmatrix}}
\renewcommand{\thefootnote}{\alph{footnote}}
\begin{document}
\begin{frontmatter}

\title{A port-Hamiltonian framework for displacement-based and rigid formation tracking}

\author[First] {Ningbo Li,} 
\author[Second] { Zhiyong Sun,}
\author[First] { Arjan van der Schaft,}
\author[First] { Jacquelien M. A. Scherpen} 
\address[First] {Jan C. Willems Center for Systems and Control, University of Groningen, The Netherlands (e-mail: ningbo.li@rug.nl, a.j.van.der.schaft@rug.nl, and j.m.a.scherpen@rug.nl).}
\address[Second]{Department of Electrical Engineering, Eindhoven University of Technology, 5600 MB Eindhoven, The Netherlands. (e-mail: z.sun@tue.nl).}

\begin{abstract}       
This paper proposes a passivity-based port-Hamiltonian (pH) framework for multi-agent displacement-based and rigid formation control and velocity tracking. The control law consists of two parts, where the internal feedback is to track the velocity and the external feedback is to achieve formation stabilization by steering variables of neighboring agents that prescribe the desired geometric shape. Regarding the external feedback, a general framework is proposed for all kinds of formations by means of the advantage that the pH model is energy-based and coordinate-free. To solve the issue that the incidence matrix is not of full column  rank over cyclic graphs, the matrix property is used to prove the convergence to the target sets for the displacement-based formation, while for rigid formations, the algebraic conditions of infinitesimal rigidity are investigated to achieve asymptotic local stability. Furthermore, the rigid formation with heterogeneous constraints is further investigated under this framework and the asymptotic local stability is proved under a mild assumption. Simulations are performed to illustrate the effectiveness of the framework.
\end{abstract}

\end{frontmatter}

\section{INTRODUCTION}

Formation control is a popular control problem within the broad class of coordination control which aims at a group behavior to achieve a prescribed geometric shape. It has been extensively investigated as motivated by many promising applications; see e.g., the survey papers, \cite{beard2001coordination}, \cite{ren2007information}, \cite{oh2015survey}, \cite{cortes2017coordinated}, \cite{liu2018survey}, and \cite{chen2019control} and the references therein. According to different formation variables that define the geometric formation shape, formation control can be classified as either displacement-based formations or rigid formations. Regarding displacement-based formation, the target geometric shape is defined by relative positions between neighboring agents. The convergence is investigated based on the multi-agent consensus theory; see e.g., \cite{olfati2004consensus}, \cite{ren2005consensus}. For rigid formations, one can employ distances, bearings, or angles, which can be regarded as partial information about the relative positions, to define a target shape. Since less information is used, more constraints should be imposed to uniquely define the formation shape, which involves the graph rigidity theory \cite{anderson2008rigid}, \cite{sun2018cooperative}. On the other hand, since the control objectives of rigid formations are represented by functions of a norm, the resulting control laws for formation stabilization are often nonlinear which imposes challenges for formation convergence analysis \cite{oh2015survey}.   

The objective of multi-agent group coordination is to achieve a prescribed group behavior using distributed feedback laws that only employ information about neighboring agents. In practice, it is often assumed that  neighboring agents have access to each other's information based on an underlying undirected graph, while the resulting closed-loop formation system will inherit the passivity properties of  agents' dynamics  \cite{van2000l2}, \cite{chopra2006passivity}, and \cite{bai2011cooperative}. Therefore, the passivity approach is a favorable tool to analyze the stability of the interconnected multi-agent system.  

In \cite{arcak2007passivity}, the notion of passivity is used for group coordination. The control objective is to steer the differences of output variables of neighboring agents to a prescribed compact set. By extending this idea to formation stabilization and velocity tracking, the passivity-based group coordination can be interpreted as the following: the controller consists of internal and external feedback, while the internal feedback renders the agents passive with respect to the desired velocity, and the external feedback steers formation shape variables of neighboring agents that prescribe the desired geometric shape, which asymptotically stabilize the formation shape variables to a given formation shape.   

Since the dynamics of the external feedback are associated with the edges, the convergence analysis is  related to the incidence matrix of the graph. For group coordination in \cite{arcak2007passivity}, a necessary condition to achieve convergence of system states to a target set is the linear independence of the columns of the incidence matrix. However, in the context of formation control, this condition is often not satisfied, for example in the displacement-based formation with cyclic graphs. Also in the case of rigid formations, the underlying graph always contains cycles, which causes the number of columns of the incidence matrix to be greater than the rank. Therefore, this should be further investigated to guarantee asymptotic stabilization to a given formation. 

The first objective of this paper is to extend the passivity framework to displacement-based formations and rigid formations modeled by cyclic graphs. To solve the above issues,
regarding the displacement-based formation with cyclic graphs, the convergence to the target sets is proved by using incidence matrix property of the underlying graph. As for the rigid formations, we have investigated the relationship between infinitesimal rigidity and the time-derivative of the Hamiltonian to prove asymptotic  local stability. In addition, regarding angle-based formation, we proposed the framework for the model of double integrators with a damping term. To the best of our knowledge, existing research, such as \cite{jing2019angle}, \cite{basiri2010distributed}, \cite{chen2020angle}, and \cite{chen2022globally} only considers the dynamics of single integrators. Furthermore, the last three papers did not give a suitable Lyapunov function for stability analysis. 

The construction of the passivity framework is accomplished by exploiting the pH formulation of passive systems, which offers a number of advantages. First, pH models are convenient for modelling interconnected dynamic systems in terms of scalability \cite{van2013port}, making them suitable to represent interconnected systems in complex networks. Since in this modeling framework the dissipation and the energy storage structure of the system are underscored, passivity-based control techniques arise as a natural option to control pH systems. Second, pH modeling is energy-based and coordinate-free, which gives a general framework for the modelling and analysis of different kinds of formations. 
Third, compared with other existing formation research, the pH approach is more suitable for complex and heterogeneous agent dynamics. Different from most existing literature where the agent is modeled as a single or double integrator and the formation shape is defined by one kind of formation variables such as position, distance, bearing, or angle, the pH approach can be further applied to the systems where the agents are modeled by more complex dynamics and the formation shapes are given by different kinds of constraints. In this paper, we have further established the passivity-based pH framework for rigid formations with mixed constraints that involve displacement, distance, bearing, and angle constraints. 

The original idea of this paper was first published in the conference paper  \cite{li2022passivity}. In this paper, we present more details about constructing this framework. In addition, we extend the framework to  angle-based formations and the rigid formations with heterogeneous constraints. The main contributions of this paper can be summarized as follows:
\begin{itemize}
 \item Our passivity approach for velocity tracking and formation control is formulated in the pH framework, which is suitable to represent complex and heterogeneous agent dynamics and is more favorable in terms of scalability. In addition, it gives a general framework for different kinds of formations due to the advantage that the model is energy-based and coordinate-free.
 \item We extend the passivity-based group coordination framework to the displacement-based and rigid formation control. To solve the issue that the incidence matrix is not of  full column rank over cyclic graphs, regarding the displacement-based formation, the convergence to the target sets is proved by employing the matrix property of the graph incidence matrix. Concerning rigid formations, we use the algebraic conditions of infinitesimal rigidity to prove asymptotic local stability.
 \item Regarding angle-based formation, compared with the existing results, we propose a general framework for the modeling of double integrators with a damping term and give a suitable Lyapunov function for stability analysis. Regarding the rigid formations with heterogeneous constraints, we define the completely infinitesimally rigid framework and establish a formation system that involves different kinds of formation geometric constraints.
\end{itemize}

The rest of the paper is structured as follows. The preliminaries are given in Section 2. The frameworks for displacement-based formation and rigid formations are presented in Sections 3 and 4, respectively. Simulations are provided in Section 5, and concluding remarks appear in Section 6.

\section{PRELIMINARIES}

\subsection{Port-Hamiltonian (pH) systems}

The pH systems theory brings together port-based modeling, geometric mechanics, and system and control theory in physical system modeling and analysis, in order to obtain a clear representation of physical processes. A standard input-state-output pH system \cite{van2014port},\cite{duindam2009modeling} is formulated as
\begin{align}
\dot{x}&=(J(x)-R(x))\frac{\partial{H}}{\partial{x}}(x)+g(x)u, \nonumber \\
y&=g^T(x)\frac{\partial{H}}{\partial{x}}(x),
\end{align}
where $x\in\mathbb{R}^n$ is the state, $u\in\mathbb{R}^m$ is the input, and $y\in\mathbb{R}^m$ is the output. Furthermore $J(x)=-J^T(x)\in\mathbb{R}^{n\times{n}}$ is the skew-symmetric interconnection matrix, $R(x)=R^T(x)\ge0\in\mathbb{R}^{n\times{n}}$ is the positive semi-definite dissipation matrix, and $H(x)$ is the Hamiltonian that equals the total   energy stored in the system. It is easy to verify that the time-derivative of $H(x)$ satisfies $\dot{H}(x)\leq y^Tu$, which leads to the passivity of the system under the assumption that $H(x)$ is bounded from below. Otherwise, the system is cyclo-passive. This passivity property is often used to prove the stability of the closed-loop system. In contrast to other modeling approaches, the pH formulation highlights the interconnection structure related to the exchange of energy. This description of systems is suited for passivity-based control, whose basic ideas are energy-shaping and control by interconnection.   

\subsection{Passivity for group coordination}

Passivity is a favorable design tool for multi-agent systems since the feedback interconnection structure ensures that the closed-loop system inherits the passivity property of its components. In addition, it also allows for modeling complex agent dynamics and is scalable to a large number of subsystems. The passivity framework for group coordination was introduced in \cite{arcak2007passivity}. 

Consider a group of $N$ agents with the topology of information exchange between these agents described by a graph $\mathcal{G}(\mathcal{V}_N,\mathcal{E}_E)$. It consists of a node set $\mathcal{V}$, where $\mathcal{V}=\{1,2,...,N\}$, and an edge set $\mathcal{E}\subseteq \mathcal{V} \times \mathcal{V}$, where $\mathcal{E}=\{e_1,e_2,...,e_E\}$. The incidence matrix $B \in \mathbb{R}^{N\times E}$ describes the relationship between the nodes and the edges, and it takes the following form:
$$
b_{ik}=
\begin{cases}
+1 & \text{if node $i$ is at the positive side of edge $k$},\\
-1 & \text{if node $i$ is at the negative side of edge $k$},\\
0  & \text{otherwise}.
\end{cases}
$$

We suppose each agent is modeled in pH form as a single point mass in $\mathbb{R}^d$. The position of each agent $i$ is denoted as $q_i \in \mathbb{R}^d$ and the corresponding momentum is defined as $p_i=m_i\dot{q}_i \in \mathbb{R}^d$, where $m_i$ is the mass of agent $i$. The dynamics of all agents are given in the compact form by
\begin{align}	\label{dyph}
\begin{split}
\matr{\dot{q}  \\ \quad \\ \dot{p}} &= \matr{0 & I_{Nd}\\ \quad \\-I_{Nd} & -{D}^r}  \matr{ \frac{\partial{H}}{\partial{q}}(q,{p}) \\ \quad \\ \frac{\partial{H}}{\partial{p}}(q,{p}) } + \matr{0\\ \quad\\I_{Nd}}u 
\\
y&=\frac{\partial{H}}{\partial{p}}(q,{p})
\end{split}
\end{align} 
where $u\in \mathbb{R}^{Nd}$ is the input, $y\in \mathbb{R}^{Nd}$ is the output, and $D^r \in\mathbb{R}^{Nd\times Nd}$ is a positive semi-definite dissipation matrix, which enables to model the viscous friction. $I_{Nd}\in \mathbb{R}^{Nd \times Nd}$ is an identity matrix. The Hamiltonian consists of the kinetic energy associated with the movement of the mass and takes the following form $H=\sum_{i=1}^{N}H_i=\frac{1}{2}\sum_{i=1}^{N}p^T_iM_i^{-1}p_i$, where $M_i=m_iI_d$.  
 
The coordination objective for the group behavior consists of two parts. One is that the velocity of each agent converges to a prescribed common value, i.e., $$\lim_{t \to \infty}|\dot{q}_i-v^*|=0.$$ The other is that the different variables associated with the edges
\begin{equation}   \label{relap}
z=(B^T\otimes I_d)q
\end{equation} 
converge to a prescribed compact set $\Xi \subset \mathbb{R}^{E\times d}$, where $\Xi =\{z_1^*,z_2^*, ..., z_E^* \}$ and $z_j^*, j \in 1,2,...,E$ is the prescribed difference variable associated with edges $j$.

To achieve the first objective, an internal feedback is designed for every single agent, which renders its dynamics passive from the designed feedback input $u_i$ to the velocity error $y_i:=\dot{q_i}-v^*$, i.e., the whole group achieves velocity tracking. 

For the second objective, an external feedback is designed based on relative difference variables associated with edges. This can be written in a compact form as $$u=-(B\otimes I_d)\phi,$$ where $\phi:\mathbb{R}^{M\times d} \to \mathbb{R}^{M\times d}$ are multivariable nonlinearities to be designed to render the prescribed compact set $\Xi \subset \mathbb{R}^{M\times d}$ invariant and asymptotically stable.

The set of desired equilibria is given by
\begin{equation}	\label{sete}
\mathcal{E}=\{(z,\xi)|\xi=0, (B\otimes I_d)\phi(z)=0, z\in \mathcal{R}(B^T\otimes I_d)\},
\end{equation} 
where $\xi=\dot{q}-\mathbf{1}_N v^*$, and $\mathcal{R}(\cdot)$ denotes the range space.

\section{DISPLACEMENT-BASED FORMATION}

We consider a group of $N$ agents which are fully actuated as described by \eqref{dyph}. The interaction among the group of agents is modeled by an undirected, connected graph $\mathcal{G}(\mathcal{V}_N,\mathcal{E}_E)$, where the nodes are associated with agents and the edges are associated to interactions between the agents.

In displacement-based formation, the measurements are the relative positions of the agents. The control law for formation stabilization is usually   based on consensus theory \cite{ren2007information}, \cite{hernandez2020consensus}, \cite{garciademarina2020maneuvering}.   The convergence result employs the Laplacian matrix property that the rank of the Laplacian matrix is $N-1$ if the underlying undirected graph is connected.   
Moreover, from the perspective of the passivity framework for group coordination, the relative difference variables in (\ref{relap}) for displacement-based formation are relative positions. Since the controller dynamics are associated with the edges, the incidence matrix is used to transform the control force from the edge space to the node space. In this regard, the convergence depends on the kernel space of the incidence matrix. 

The underlying graph of a network does not contain cycles if and only if the columns of incidence matrix $B$ are linearly independent. Equivalently, the kernel space of $B\otimes I_d$ is trivial, and $(B\otimes I_d)\phi(z)=0$ implies $\phi(z)=0$, and hence $z\in \Xi$. If the graph contains cycles, the columns of the incidence matrix $B$ are linearly dependent. In this case, the kernel space of $B\otimes I_d$ is not trivial; therefore, the formation shape variables may not converge to a given formation shape. On the other hand, if a group of agents is connected via a cyclic graph, rather than an acyclic tree graph, the cyclic graph structure enhances the robustness of the formation system: if one of the agents fails, the remaining graph is still connected as an acyclic or cyclic graph, and the whole system still works. Therefore, cyclic graphs are important in displacement-based formation.

To address this issue, we design a passivity approach for displacement-based formation stabilization and velocity tracking for a group of fully actuated agents modeled in (\ref{dyph}).

In terms of velocity tracking, the following generalized canonical coordinate transformations \cite{fujimoto2003trajectory} are introduced for the pH model in order to derive the error dynamics
\begin{equation*}
\matr{\bar{q}_i(t) \\ \bar{p}_i(t)} = \matr{{q}_i-v^*t \\ p_i-M_iv^*},
\end{equation*} 
where $v^* \in \mathbb{R}^d$ is the prescribed desired velocity. The Hamiltonian for velocity tracking is given as $$H^v=\sum_{i=1}^{N}(H_i+U_i)=\frac{1}{2}\sum_{i=1}^{N}({p}^T_iM_i^{-1}{p}_i-2p_i^Tv_i^*+{v_i^*}^TM_iv_i^*),$$  
where $H_i$ is the kinetic energy for each agent, $U_i$ is a fictitious potential. When $\dot{q}_i$ for all $i=1,2,...,N$ converges to $v^*$, $H^v$ has the minimum value.

To eliminate the tracking error, the corresponding control law consisting of two terms is given by 
\begin{align}	\label{uv}
\begin{split}
u^v_{i}=-\frac{\partial U_i}{\partial p_i}-\frac{\partial H^v_i}{\partial {p}_i}=-D^r_iv^*-D_i^tM_i^{-1}\bar{p}_i.
\end{split}
\end{align} 
where $D_i^t \in \mathbb{R}^{d \times d}$ is a positive semi-definite dissipation matrix. The first term, $-\frac{\partial U_i}{\partial p_i}$, ensures the velocity converging to the desired value, while the second term, $-D_i^t\frac{\partial H^v_i}{\partial {p}_i}$, improves the transient performance of the convergence. 

\begin{proposition}
Consider a group of agents modeled   in (\ref{dyph}). Using the control law (\ref{uv}), the system converges to the desired velocity $v^*$.
\end{proposition}

The proof directly follows by taking $H^v$ as the Lyapunov function, and we omit the details here.

For formation stabilization, the displacement error of edge $j$ is defined as $\bar{z}_j=z_j-z_j^*$, where $z_j^*$ is the desired displacement of edge $j, j=1,2,..., E$. To achieve the control objectives, virtual couplings are assigned among the neighboring agents \cite{van2014port}. The corresponding Hamiltonian is equal to the virtual potential energy of the edges and is given by $H^f=\frac{1}{2}\sum_{j=1}^{E}\bar{z}_j^T\bar{z}_j$. For each edge, when $z_j$ converges to $z_j^*$, $H_i^f$ attains its minimum value. Therefore, when $z\in \Xi$, $H^f$ has the minimum value. Note that since the positions of the agents are not controlled, the terms that contain $q$ are not included in the virtual potential energy.

The dynamics of the formation control system associated with the edges of the graph are given by 
\begin{align}	\label{vc}
\begin{split}
\dot{z}_j&=\omega^z_j,\\
\tau_j^z&=\frac{\partial H^f_j}{\partial z_j}+D^f_j\omega^z_j,
\end{split}
\end{align} 
where $\omega^z_j \in \mathbb{R}^d$ and $\tau^z_j \in \mathbb{R}^d$ are the input and output of the controller, respectively, and $D^f_j \in \mathbb{R}^{d\times d} \geq 0$ is the dissipation matrix. 

The original systems are described by the dynamics of the agents, while the controllers are designed by the dynamics on the edges. Therefore, the incidence matrix $B$ is used to establish the interconnection of the original system and the controllers, in the compact form given by
\begin{align}	\label{uy}
\begin{split}
u^f=-(B\otimes I_d)\tau^z,\\
\omega^z=(B^T\otimes I_d)y.
\end{split}
\end{align} 

According to (\ref{vc}) and (\ref{uy}), the control law for formation stabilization follows directly as
\begin{align}
    u^f=-(B\otimes I_d)\bar{z}-(B\otimes I_d)D^f(B^T\otimes I_d)M^{-1}p,
\end{align}
where the spring term, $-(B\otimes I_d)\bar{z}$, is to determine the stability of the formation shape, while the damping term, $-(B^T\otimes I_d)D^f(B\otimes I_d)M^{-1}p$, is to improve the transient performance of convergence.

In general, the total control input is given by
\begin{align} \label{uc}
u^c=u^v+u^f.    
\end{align}
The first term $u^v$ is an internal control law for velocity tracking, where the agents only need the information about their own. The second term $u^f$ is an external control law for formation control, where the agents need information about their neighbors.   

The total Hamiltonian is defined as $H^c=H^v+H^f$. In this case, the closed-loop system is given by
\begin{align}	\label{clsys}
\begin{split}
\matr{\dot{\bar{q}}  \\ \quad \\ \dot{\bar{p}} \\ \quad \\ \dot{\bar{z}} } = \matr{0 & I_{Nd} & 0 \\ \quad \\-I_{Nd} & -{D}^F & -B\otimes I_d \\ \quad \\ 0 & {B}^T\otimes I_d  & 0}  \matr{ \frac{\partial{H^c}}{\partial{\bar{q}}}(\bar{q},\bar{p}) \\ \quad \\ \frac{\partial{H^c}}{\partial{\bar{p}}}(\bar{q},\bar{p}) \\ \quad \\ \frac{\partial{H^c}}{\partial{\bar{z}}}(\bar{q},\bar{p}) }  
\end{split}
\end{align} 
where $D^F=D^r+D^t+(B\otimes I_d)D^f(B^T\otimes I_d)$.

We obtain the following theorem on the displacement-based formation and velocity tracking.
\begin{theorem} \label{theo1}
Consider a group of agents modeled in port-Hamiltonian form as in (\ref{dyph}), and assume that the graph is undirected and connected. Then the control law (\ref{uc}) achieves the desired formation while each agent tracks the desired velocity. 
\end{theorem}

\begin{proof}
Take the Hamiltonian $H^c$ as the Lyapunov function. It follows that $H^c(\bar{q},\bar{p},\bar{z})\ge 0$. Taking the time-derivative of $H^c$, we have $$\dot{H}^c=-\frac{\partial^T H^c}{\partial \bar{p}}(D^r+D^t+(B\otimes I_d)D^f(B^T \otimes I_d))\frac{\partial H^c}{\partial \bar{p}}.$$ It follows that $\dot{H}^c \leq 0$. Furthermore, by invoking LaSalle's Invariance Principle the system converges to the largest invariant set $\{p|\bar{p}=0\}$. It follows that $\dot{\bar{p}}=0$. Substituting $\bar{p}=0, \dot{\bar{p}}=0$ into (\ref{clsys}) gives $$-(B\otimes I_d) \frac{\partial{H^c}}{\partial{\bar{z}}}=-(B\otimes I_d)\bar{z}=0.$$

Next, we consider two different cases of the underlying connected graphs. 
\begin{itemize}
 \item \textbf{Acyclic graph}. If the connected graph is acyclic, the columns of the incidence matrix $B$ are linearly independent. It follows that the kernel of $B$ is a null set. Therefore, $\bar{z}=0$, which means that the agents achieve the desired formation shape, thus completing the proof.
 \item \textbf{Cyclic graph}. For a cyclic graph, the columns of incidence matrix $B\otimes I_d$ are linearly dependent, which implies that the kernel of $B\otimes I_d$ is not a null set anymore. However, according to the definition of difference variables $z=(B\otimes I_d)^Tq$, it follows that $\bar{z}$ is in the range space of $B^T\otimes I_d$. Since the range space of $B^T\otimes I_d$ is orthogonal to the kernel of $B\otimes I_d$, it implies $\bar{z}=0$, thus completing the proof.
\end{itemize} 
The proof of the theorem statement is thus completed. 
\end{proof}

\section{RIGID FORMATIONS}
In the case of rigid formations, the relative difference variables under consideration are distances, bearings, angles, and combinations of all measurements. When the constraints of the formation involve more geometric variables than displacement, the graph usually contains cycles, which implies that the kernel of the incidence matrix is not zero anymore. To address this issue, we propose the passivity approach in pH form for the modeling and stabilization analysis of these rigid formations. 

The controller for rigid formation tracking consists of two parts, where the internal feedback is for velocity tracking and the external feedback is for formation stabilization. Since there is no coupling between the velocity and position dynamics, the same control law as in (\ref{uv}) can still be used for velocity tracking. Therefore, we only consider the formation stabilization in this section. 

\subsection{Distance-based formation}
In order to formulate the concept of rigidity more clearly, the definition of the framework in formation control is given as follows \cite{ahn2020formation}.

\begin{definition}
Given a graph $\mathcal{G}(\mathcal{V},\mathcal{E})$, an associated framework  $f_q$ is a realization of the underlying  graph with node coordinate variables $q$, i.e. $f_q=(\mathcal{G},q)$.
\end{definition}

The distance of edge $j$ that associates agents $i$ and $k$ is defined as $$||z_j||=||q_i-q_k||.$$ 

In terms of the distance rigidity, the potential function is defined as $$h^d=[||z_1||^2, ||z_2||^2,...,||z_E||^2]^T.$$ 

The time-derivative of $h^d$ is given as
\begin{align} \label{drm}
 \dot{h}^d=\frac{\partial^T h^d}{\partial q}\dot{q}= {\rm blkdiag}(z_1^T,z_2^T,...,z_E^T)(B^T \otimes I_d)\dot{q},   
\end{align}
where $R^d=\frac{\partial^T h^d}{\partial q} \in \mathbb{R}^{E\times Nd}$ is defined as the distance rigidity matrix. For more details on distance rigidity, we refer the readers to \cite{anderson2008rigid}, \cite{sun2018cooperative}, \cite{hendrickson1992conditions}. We quote the following {\it Lemma \ref{lemd1}} from \cite{hendrickson1992conditions} and {\it Lemma \ref{lemd2}} from \cite{sun2018cooperative}.
\begin{lemma} \label{lemd1}
Assume the number of agents $N$ is greater than the dimension $d$ of the ambient space of the framework. A framework $f=(\mathcal{G}_N(\mathcal{V}_N,\mathcal{E}_E),q)$ is infinitesimally distance rigid (IDR) in $\mathbb{R}^d$ if and only if $rank(R_d)=dN-d(d+1)/2$.
\end{lemma}
\begin{lemma} \label{lemd2}
 If a framework $f=(\mathcal{G}_N(\mathcal{V}_N,\mathcal{E}_E),q)$ is minimally and infinitesimally distance rigid (MIDR) in $\mathbb{R}^d$, then the matrix $R^d(q)R^d(q)^T$ is positive definite.
\end{lemma}

To make the distance of each edge go to the desired value, the Hamiltonian for formation stabilization is given as $$H^d=\sum_{j=1}^{E}H^d_j=\frac{1}{4}\sum_{j=1}^{E}(e^d_j)^2=\frac{1}{4}\sum_{j=1}^{E}(||z_j||^2-||z_j^*||^2)^2.$$

For the dynamics of the controller, we consider the virtual spring and damping couplings associated with the edges, which are given by
\begin{align}	\label{dcd}
\begin{split}
\dot e^d_j&=\omega^d_j,\\
\tau_j^d&=\frac{\partial H^d_j}{\partial e^d_j}+D^d_j\omega^d_j,
\end{split}
\end{align} 
where $\omega^d_j \in \mathbb{R}$ and $\tau_j^d \in \mathbb{R}$ are the input and output of the controller, respectively, and $D^d_j \in \mathbb{R} \geq 0$ is the dissipation constant. Note that the output of the controller is in distance space $\mathbb{R}$, to transform it to $\mathbb{R}^d$, the mapping Jacobian is given by
$$J^d_j=\frac{\partial e^d_j}{\partial z_j}=z_j$$

Establishing the negative feedback connection of the controller system and the original system, the corresponding control law is derived as $$u^d_j=-J^d_j\tau_j^d=-z_j(e^d_j+D^d_j\dot{e}^d_j).$$

Furthermore, interconnecting the whole network by using the incidence matrix $B$ of the underlying graph, the resulting control law in $\mathbb{R}^{N \times d}$ is proposed as 
\begin{align}	\label{cld}
\begin{split}
u^d=-\underbrace{(B\otimes I_d)({\rm blkdiag} (z_1^T,z_2^T,...,z_E^T))^T}_{{R^d}^T} (e^d+D^d\dot{e}^d) 
\end{split}
\end{align} 

The total Hamiltonian is defined as $H^D=H^v+H^d$. In this case, the closed-loop system is given by

\begin{align}	\label{clsd}
\begin{split}
\matr{\dot{\bar{q}}  \\ \quad \\ \dot{\bar{p}} \\ \quad \\ \dot{e}^d } =& \matr{0 & I_{Nd} & 0 \\ \quad \\-I_{Nd} & -{D}^D & -(B\otimes I_d){\Omega^d}^T \\ \quad \\ 0 & \Omega^d({B}^T\otimes I_d)  & 0}  \\ & \cdot \matr{ \frac{\partial{H^D}}{\partial{\bar{q}}}(\bar{q},\bar{p}) \\ \quad \\ \frac{\partial{H^D}}{\partial{\bar{p}}}(\bar{q},\bar{p}) \\ \quad \\ \frac{\partial{H^D}}{\partial{e^d}}(\bar{q},\bar{p}) },  
\end{split}
\end{align} 
where $\Omega^d={\rm blkdiag}(z_1^T,z_2^T,...,z_E^T)$, $D^D=D^r+D^t+(B\otimes I_d){\Omega^d}^TD^d\Omega^d({B}^T\otimes I_d)$, and $D^d={\rm diag} \{D^d_1, D^d_2, ..., D^d_E\}$.

Furthermore, we obtain the following theorem.

\begin{theorem} \label{theo2}
Consider a group of agents modeled in port-Hamiltonian form as described by \eqref{dyph}. If the desired framework $f=(\mathcal{G}_N(\mathcal{V}_N,\mathcal{E}_E),q^*)$ is IDR, then  using the control law $u^v$ (\ref{uv})+$u^d$ (\ref{cld}), the distance functions with associated edges converge to the desired values locally and asymptotically,  and each agent tracks the desired velocity. 
\end{theorem}

\begin{proof}
Take the Hamiltonian $H^D$ as a Lyapunov candidate. It follows that $H^D\geq0$. Taking the time-derivative of the Hamiltonian, we have $$\dot{H}^D=-\frac{\partial^T H^D}{\partial \bar{p}}D^D \frac{\partial H^D}{\partial \bar{p}}.$$
It follows that $\dot{H}^D \leq 0$. Furthermore, by invoking LaSalle's Invariance Principle, the system converges to the largest invariant set $\{p|\bar{p}=0\}$. It follows that $\dot{\bar{p}}=0$. Substituting $\bar{p}=0, \dot{\bar{p}}=0$ into (\ref{clsd}) gives
\begin{align} \label{botd}
-(B\otimes I_d){\Omega^d}^T \frac{\partial{H^D}}{\partial{e^d}}=-(B\otimes I_d){\Omega^d}^Te^d=-{R^d}^Te^d=0.  
\end{align}

Next, we consider $R^d, e^d$ corresponding to the MIDR framework. By invoking {\it Lemma \ref{lemd2}}, we have that $R^d{R^d}^T$ is positive definite. It follows that ${R^d}^T$ is full column rank, implying $e_d$ corresponding to the MIDR framework is zero on the invariant set.  

The IDR framework can be obtained by adding edges to the MIDR framework, in such a way that the number of the edges, $E$, is larger than the rank, $Nd-d(d+1)/2$. We have proved that the $Nd-d(d+1)/2$ edges making up the MIDR framework converge to the desired values.  Furthermore, an IDR framework can be decomposed into a MIDR sub-framework and an additional sub-framework with the remaining edges. The additional distance constraints imposed by the sub-framework are redundancy \cite{anderson2008rigid}. Therefore, when the $Nd-d(d+1)/2$ edges making up the MIDR framework converge to the desired values, all $E$ edges in an IDR framework converge to the desired values \cite{sun2016exponential}; i.e., $e_d$ corresponding to the IDR framework is zero on the invariant set, thus completing the proof.  
\end{proof}

\subsection{Bearing-based formation}
The definition for the bearing of edge $j$ that associates agents $i$ and $k$ is given by $$s_j=\frac{q_i-q_k}{||q_i-q_k||}.$$ 

In terms of bearing rigidity, the bearing potential function is defined as $$h^b=[s_1^T, s_2^T,...,s_E^T]^T.$$ 

The time-derivative of $h^b$ is given as
\begin{align} \label{brm}
\begin{split}
 \dot{h}^b&=\frac{\partial^T h^b}{\partial q}\dot{q}\\&=  {\rm blkdiag}(P_{s_1}/||z_1||,P_{s_2}/||z_2||,...,P_{s_E}/||z_E||)({B^b}^T \otimes I_d)\dot{q},  
\end{split}
\end{align}
where $P_{s_1}=I_d-{s_1}{s_1^T}$ is the orthogonal projection matrix which projects vectors onto the orthogonal complement of $s_1$, and $R^b=\frac{\partial^T h^b}{\partial q}$ is defined as the bearing rigidity matrix. For more details of bearing rigidity, we refer the readers to \cite{eren2012formation}, \cite{zhao2015bearing}.

To make the bearing of each edge go to the desired value, the Hamiltonian for formation stabilization is given as $$H^b=\sum_{j=1}^{E}H^b_j=\frac{1}{2}\sum_{j=1}^{E}(e^b_j)^Te^b_j=\frac{1}{2}\sum_{j=1}^{E}(s_j-s_j^*)^T(s_j-s_j^*).$$

For the dynamics of the controller, we consider the virtual spring and damping couplings associated with the edges, which are given by
\begin{align}	\label{dcb}
\begin{split}
\dot e^b_j&=\omega^b_j,\\
\tau_j^b&=\frac{\partial H^b_j}{\partial e^b_j}+D^b_j\omega^b_j,
\end{split}
\end{align} 
where $\omega^d_j \in \mathbb{R}^d$ and $\tau_j^d \in \mathbb{R}^d$ are the input and output of the controller, respectively, and $D^b_j \in \mathbb{R}^{d \times d} \geq 0$ is the diagonal dissipation matrix. Note that the output of the controller is in bearing space, to transform it to $\mathbb{R}^d$, the mapping Jacobian is given by
$$J^b_j=\frac{\partial e^b_j}{\partial z_j}=\frac{P_{s_j}}{||z_j||}.$$

Establishing the negative feedback connection of the controller system and original system, the corresponding control law is derived as $$u^b_j=-J^b_j\tau_j^b=-\frac{P_{s_j}}{||z_j||}(e^b_j+D^b_j\dot{e}^b_j).$$

Furthermore, interconnecting the whole network by using the incidence matrix $B$ of the underlying graph, the resulting control law in $\mathbb{R}^{N \times d}$ is proposed as 
\begin{align}	\label{clb}
\begin{split}
u^b=-\underbrace{(B\otimes I_d) {\rm blkdiag}(P_{s_1}/||z_1||,...,P_{s_E}/||z_E||)}_{{R^b}^T}(e^b+D^b\dot{e}^b)
\end{split}
\end{align} 

The total Hamiltonian is defined as $H^B=H^v+H^b$. In this case, the closed-loop system in the compact form is given by

\begin{align}	\label{clsb}
\begin{split}
\matr{\dot{\bar{q}}  \\ \quad \\ \dot{\bar{p}} \\ \quad \\ \dot{e}^b } =& \matr{0 & I_{Nd} & 0 \\ \quad \\-I_{Nd} & -{D}^B & -(B\otimes I_d){\Omega^b}^T \\ \quad \\ 0 & \Omega^b({B}^T\otimes I_d)  & 0} \\ &\cdot  \matr{ \frac{\partial{H^B}}{\partial{\bar{q}}}(\bar{q},\bar{p}) \\ \quad \\  \frac{\partial{H^B}}{\partial{\bar{p}}}(\bar{q},\bar{p}) \\ \quad \\ \frac{\partial{H^B}}{\partial{e^b}}(\bar{q},\bar{p}) }, 
\end{split}
\end{align} 
where $\Omega^b={\rm blkdiag}(P_{s_1}/||z_1||,P_{s_2}/||z_2||,...,P_{s_E}/||z_E||)$, $D^B=D^r+D^t+(B\otimes I_d){\Omega^b}^TD^b\Omega^b({B}^T\otimes I_d)$, and $D^b={\rm blockdiag} \{D^b_1, D^b_2, ..., D^b_E\}$.

Furthermore, according to the results in \cite{zhao2015bearing}, we first present the following Lemma.
\begin{lemma} \label{lembd}
If the framework $f=(\mathcal{G}_N(\mathcal{V}_N,\mathcal{E}_E),q^*)$ is infinitesimally bearing rigid, then the dynamics $\dot{\delta}(t)=(B\otimes I_d)$ ${\rm blkdiag}(P_{s_1},...,P_{s_E})s^*$ has two equilibria $\delta_1=0$ and $\delta_2=-2q^*-(\mathbf{1}\otimes ((\mathbf{1}\otimes I_d)^Tq/N))$, where $\delta_1$ is asymptotically stable and $\delta_2$ is unstable.
\end{lemma}

We obtain the following theorem on bearing-based formation and velocity tracking.

\begin{theorem} \label{theo3}
Consider a group of agents modeled in port-Hamiltonian form as described by \eqref{dyph}. If the framework $f=(\mathcal{G}_N(\mathcal{V}_N,\mathcal{E}_M),q^*)$ is infinitesimally bearing rigid, then using the control law $u^v$ (\ref{uv})+$u^b$ (\ref{clb}), each agent tracks the desired velocity and the bearings of the edges achieve the asymptotic stability of the desired sets  except for the initial condition where $s(0)=-s^*$. 
\end{theorem}

\begin{proof}
Take the Hamiltonian $H^B$ as a Lyapunov candidate. It follows that $H^B\geq0$. Taking the time derivative of the Hamiltonian, we have $$\dot{H}^B=-\frac{\partial^T H^B}{\partial \bar{p}}D^B \frac{\partial H^B}{\partial \bar{p}}.$$ 
It follows that $\dot{H}^B \leq 0$. Furthermore, invoking LaSalle Invariance Principle gives that the system converges to the largest invariant set where $\{p|\bar{p}=0\}$. It follows $\dot{\bar{p}}=0$. Substituting $\bar{p}=0, \dot{\bar{p}}=0$ into (\ref{clsb}) gives
\begin{align} \label{dotb}
-(B\otimes I_d){\Omega^b}^T \frac{\partial{H^B}}{\partial{e^b}}=-(B\otimes I_d){\Omega^b}^Te^b=0.  
\end{align}

Since $P_{s_j}s_j=0$ and $e_j^b=s_j-s_j^*$, it implies $-(B\otimes I_d){D^b}^Ts^*=0$. 

According to Lemma \ref{lembd}, in terms of bearings, $\delta_1$ corresponds to $s=s^*$ and $\delta_2$ corresponds to $s=-s^*$. Since the distance term is always a non-negative scalar, the convergence of the system is not affected by the evolution of the distance term. Therefore, the invariant sets $\{s|-(B\otimes I_d){\rm blkdiag} (P_{s_1},...,P_{s_E})^Ts^*=0\}$ and $\{s|-(B\otimes I_d){\Omega^b}^Ts^*=0\}$ are the same. On this invariant set, there are two elements $s=s^*$ and $s=-s^*$. Hence, in terms of the closed-loop system, $e^b=0$ is an asymptotically stable equilibrium and $e^b=-2s^*$ is an unstable equilibrium. Therefore, the almost global convergence of $e^b=0$ with all initial conditions except $s(0)=-s^*$ can be obtained. This completes the proof.
\end{proof}

\begin{figure} 
\begin{center} 
\label{angle_diagram}
\includegraphics[width=4.5cm]{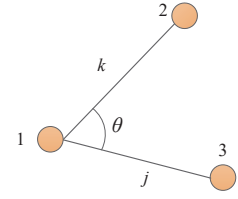}   
\caption{Angle $\theta$} 
\end{center}
\end{figure}

\subsection{Angle-based formation}
Regarding the angle-based formation, we consider the agent dynamics (\ref{dyph}) of dimension 2. Similarly, the controller for the angle-based formation consists of two parts, where the internal feedback is for velocity tracking and the external feedback is for formation stabilization. The same control law as in (\ref{uv}) is used for velocity tracking. Next, we consider formation stabilization. Since the inner constraints of the agents in formation are given by angles and the cosine of the angle is monotone with respect to the angle, we use the cosine of the angle to represent the angle measurement, which can be easily calculated by bearing measurements. For the angle $\theta$ formed by agents $1,2,3$ as shown in Fig. 1, it is given by 
\begin{equation}	
\cos \theta= s_k^Ts_j.
\end{equation} 

The time-evolution of the angle $\theta$ can be derived as:
\begin{equation} \label{costheta}
\begin{split}
\frac{\rm{d}(\cos{\theta})}{{\rm d} t}&=(\frac{P_{s_k}}{||z_k||}\dot{z}_k)^Ts_j+s_k^T(\frac{P_{s_j}}{||z_j||}\dot{z}_j)
              \\&=-(s_j^T\frac{P_{s_k}}{||z_k||}+s_k^T\frac{P_{s_j}}{||z_j||})\dot{q}_1+s_j^T\frac{P_{s_k}}{||z_k||}\dot{q}_2+s_k^T\frac{P_{s_j}}{||z_j||}\dot{q}_3.
\end{split}
\end{equation}

Furthermore, we give the following definitions
\begin{equation} \label{l1}
\begin{split}
&L_{\theta 1}:=-(s_j^T\frac{P_{s_k}}{||z_k||}+s_k^T\frac{P_{s_j}}{||z_j||}),\\
&L_{\theta 2}:=s_j^T\frac{P_{s_k}}{||z_k||},\\
&L_{\theta 3}:=s_k^T\frac{P_{s_j}}{||z_j||}.
\end{split}
\end{equation}
Therefore, (\ref{costheta}) can be rewrite as 
\begin{equation} 
\begin{split}
\frac{\rm{d}(\cos{\theta})}{{\rm d} t}=L_{\theta 1}\dot{q}_1+L_{\theta 2}\dot{q}_2+L_{\theta 3}\dot{q}_3.
\end{split}
\end{equation}

We consider a particular class of underlying graphs $\mathcal{G}_N(\mathcal{V}_N,\mathcal{E})$ where the framework $f=(\mathcal{G}_N,q)$ is angle rigid. For more details about angle rigidity, one may refer to \cite{jing2019angle}, \cite{chen2020angle}, and \cite{chen2017global}.

In terms of the angle rigidity, the angle potential function is defined as $$h^a=[\cos{\theta_1}, \cos{\theta_2}, ..., \cos{\phi_M}]^T.$$ 
Note that $h^a: \mathbb{R}^{2N} \rightarrow \mathbb{R}^{M}$. The time-derivative of $h^a$ is given as
\begin{align} \label{ar}
 \dot{h}^a=\frac{\partial^T h^a}{\partial q}\dot{q}.
\end{align}
$\frac{\partial^T h^a}{\partial q}$ is defined as angle rigidity matrix $R^a$, which can be given by angle Jacobian of each angle $\theta_l, l \in \{1, 2, ..., M\}$ with respect to the three related agents. For example, let $\theta_l$ be formed by agents $n_1, n_2, n_3$, then the expression of the angle rigidity matrix $R^a$ can be shown as follows:
\begin{equation} \label{arm}
\begin{split}
\bordermatrix{
         & 1 &\cdots & n_1 &\cdots & n_2 &\cdots & n_3 &\cdots & N \cr 
\theta_1 &\cdots &\cdots &\cdots &\cdots &\cdots &\cdots &\cdots &\cdots &\cdots \cr
\theta_2 &\cdots &\cdots &\cdots &\cdots &\cdots &\cdots &\cdots &\cdots &\cdots \cr
\cdots   &\cdots &\cdots &\cdots &\cdots &\cdots &\cdots &\cdots &\cdots &\cdots \cr
\theta_l &\cdots &\cdots &L_{\theta_l n_1} &\cdots &L_{\theta_l n_2} &\cdots &L_{\theta_l n_3} &\cdots &\cdots \cr
\cdots   &\cdots &\cdots &\cdots &\cdots &\cdots &\cdots &\cdots &\cdots &\cdots \cr
\theta_M   &\cdots &\cdots &\cdots &\cdots &\cdots &\cdots &\cdots &\cdots &\cdots
}.
 \end{split}
\end{equation}
Regarding the property of $R^a$, we give the following lemma.

\begin{lemma} \label{lem4}
 If a framework $f=(\mathcal{G}_N(\mathcal{V}_N,\mathcal{E}_E),q)$ is minimally and infinitesimally angle rigid (MIAR) in $\mathbb{R}^d$, then the corresponding angle rigidity matrix $R^a(q)R^a(q)^T$ is positive definite.
\end{lemma}

A framework is infinitesimally angle rigid if all its continuous infinitesimally motions are trivial in terms of angle constraints. A framework $f$ is minimally angle rigid if its angle rigidity cannot be guaranteed when removing one angle constraint. The proof of Lemma \ref{lem4} is obvious from the definition of $R^a$.

In the angle-based formation, the Hamiltonian for formation stabilization is given as $$H^a=\sum_{l=1}^{M}H^a_l=\frac{1}{2}\sum_{l=1}^{M}{e^a_{\theta_l}}^2,$$
where $e^a_{\theta_l}=\cos{\theta_l}-\cos{\theta_l}^*$.

For the dynamics of the controllers of agents $n_1, n_2, n_3$ forming the triangle $l$, we consider the virtual spring and damping couplings associated with the angles. The dynamics are given by
which are given by
\begin{align}	\label{dca}
\begin{split}
\dot e^a_{\theta_l}&=\omega^a_{\theta_l},\\
\tau_{\theta_l}^a&=\frac{\partial H^a_{l}}{\partial e^a_{\theta_l}}+D^a_{\theta_l}\omega^a_{\theta_l},
\end{split}
\end{align}
where $\omega^a_{\theta_l} \in \mathbb{R}$ and $\tau_{\theta_l}^a \in \mathbb{R}$ are the input and output of the controller, respectively, and $D^a_{l\theta} \in \mathbb{R} \geq 0$ is the dissipation constant. Note that the output of the controller is in angle space, to transform it to $\mathbb{R}^2$, the mapping Jacobian is given by
$$J^a_{\theta_l}=\frac{\partial e^a_{\theta_l}}{\partial q}=\matr{L^T_{\theta_l n_1} & L^T_{\theta_l n_2} & L^T_{\theta_l n_3}} \in \mathbb{R}^{2 \times 3},$$
which actually corresponds to the $l$th row of the angle rigidity matrix.

Establishing the negative feedback connection of the controller system and original system and interconnecting agents $n_1, n_2, n_3$ by using the mapping Jacobian, the resulting control law is given by
\begin{align}	\label{caijk}
\begin{split}
u^a_{n_1}=-L^T_{\theta_l n_1}(e_{\theta_l}+d^a_{l\theta}\dot e_{\theta_l}), \\
u^a_{n_2}=-L^T_{\theta_l n_2} (e_{\theta_l}+d^a_{l\theta}\dot e_{\theta_l}),\\
u^a_{n_3}=-L^T_{\theta_l n_3} (e_{\theta_l}+d^a_{l\theta}\dot e_{\theta_l}).
\end{split}
\end{align} 

Considering a angle angle rigid framework $f=(\mathcal{G}_N(\mathcal{V}_N,\mathcal{E}_E),q)$ , the overall control law is derived as 
\begin{equation} \label{cla}
\begin{split}
u^a_n &= \sum_{l \in \mathcal{N}_n}({L}^T_{\theta_l n}(e_{\theta_l}++d^a_{\theta_l}\dot e_{\theta_l})), \quad n \in \{1,2,...,N\},
\end{split}
\end{equation}
where the set $\mathcal{N}_n$ contains all the angle constraints $l$ that involve agent $n$.

The total Hamiltonian is defined as $H^A=H^v+H^a$. In this case, the closed-loop system in the compact form is given by
\begin{align}	\label{clsa}
\begin{split}
\matr{\dot{\bar{q}}  \\ \quad \\ \dot{\bar{p}} \\ \quad \\ \dot{e}^a } =& \matr{0 & I_{2N} & 0 \\ \quad \\-I_{2N} & -{D}^A & -{R^a}^T \\ \quad \\ 0 & R^a  & 0} \matr{ \frac{\partial{H^A}}{\partial{\bar{q}}}(\bar{q},\bar{p}) \\ \quad \\  \frac{\partial{H^A}}{\partial{\bar{p}}}(\bar{q},\bar{p}) \\ \quad \\ \frac{\partial{H^A}}{\partial{e^a}}(\bar{q},\bar{p}) },
\end{split}
\end{align} 
where $e^a= \matr{e_{\theta_1} & e_{\theta_2} & ... & e_{\phi_M}}^T \in \mathbb{R}^{2M}$, and  $D^A=D^r+D^t+{R^a}^TD^aR^a$, $D^a={\rm diag} \{D^a_{\theta_1}, D^a_{\theta_2}, ..., D^a_{\theta_M}\}$.

The main result is given by the following theorem.
\begin{theorem} \label{theo4}
Consider a network of the agents modeled as in (\ref{dyph}) and connected by a triangulated Laman graph $\mathcal{G}_n(\mathcal{V}_n,\mathcal{E})$ with $M$ triangles. Moreover, assume that any three agents forming a triangle are not collinear and no agents are coincident at the initial time. If the framework $f=(\mathcal{G}_N(\mathcal{V}_N,\mathcal{E}_M),q^*)$ is infinitesimally angle rigid, then using the control law $u^v$ (\ref{uv})+$u^a_n$ (\ref{cla}), the group of agents locally and asymptotically converge to the desired formation constrained by the angle constraints, and each agent tracks the desired velocity.  
\end{theorem}

\begin{proof}
Take the following Hamiltonian as a candidate Lyapunov function 
\begin{equation} \label{aH}
\begin{split}
H^A=&\frac{1}{2}\sum_{n=1}^N \frac{1}{m_n}\bar{p}_n^T\bar{p}_n
+ \frac{1}{2}\sum_{l=1}^Me_{\theta_l}^2.
\end{split}
\end{equation}
It follows that $H^A$ is positive definite. Now we consider the time derivative of (\ref{aH}). For simplicity, we omit the process and only give the results as follows
\begin{equation} \label{daH}
\begin{split}
\dot{H}^A=&-\frac{\partial^T H^A}{\partial \bar{p}}(D^r+D^t)\frac{\partial H^A}{\partial \bar{p}} -\sum_{l=1}^{M}D_{\theta_l}\dot{e}_{\theta_l}^2  \leq0.
\end{split}
\end{equation}
By invoking LaSalle's invariance principle, we get that the trajectories of the closed-loop system converge to the largest invariant set where $\{(p, e_{\theta})|\bar{p}=0, \dot{e}_{\theta}=0\}$. On this set $\bar{p}_n=0$ for all $n \in \{1, 2, ... ,N\}$ and $\dot{e}_{\theta_l}=0$ for all $l \in \{1,2,...,M\}$.

Furthermore, we derive that $\dot{p}_n=0$ for all $n \in \{1, 2, ... ,N\}$ due to $\bar{p}_n=0$. Substituting $\dot{p}_n=0$ into second row of (\ref{clsa}) it follows that 
\begin{equation} \label{Uneq0}
u^a_n=0, \quad n \in \{1, 2, ... ,N\}
\end{equation}
on this invariant set. Substituting $\dot{e}_{\theta_l}=0$ for all $l \in \{1,2,...,M\}$ into (\ref{Uneq0}) and writing the result in compact form, we get
\begin{equation} 
\begin{split}
0 = {R^a}^T e_{\theta}.
\end{split}
\end{equation}
Next, we first consider $R^a, e_{\theta}$ corresponding to the MIAR framework. According to $\textit{Lemma \ref{lem4}}$, ${R^a}^T$ is full column rank. Therefore, $e_{\theta}$ corresponding to the MIAR framework is zero on the invariant set. The angle rigid framework can be obtained by adding more angle constraints to the MIDR framework. The two frameworks define the same angle-based formation, i.e., the adding angle constraints in the angle rigid framework are redundancy.
Therefore, when the angle constraints in the MIAR framework converge to the desired values, all the angle constraints in the angle rigid framework converge to the desired values as well, i.e., $e_{\theta}$ corresponding to the angle rigid framework is zero on the invariant set, thus completing the proof.  
\end{proof}

\subsection{Formation control with heterogeneous constraints}
In engineering practice, to facilitate completing a certain task, the geometric shape of a formation may be defined by the combinations of different variables, such as displacements, distances, bearings, and angles. According to different formation constraint variables, the formation can achieve different maneuvering motions while respecting shape constraints. For example,  the distance-based formation is invariant to translation and rotation, the bearing-based formation is invariant to translation and scaling, and the angle-based formation is invariant to translation, rotation, and scaling. In this section, we consider all the constraint variables together, while the resulting formation is only invariant to translation, i.e., tracking the desired trajectory. 

We consider a group of agents in $\mathbb{R}^2$ modeled as in (\ref{dyph}) with heterogeneous constraints to achieve formation tracking. We assume the underlying graph is a triangulated graph. The geometric shape of the formation is given by combinations of displacements, distances, bearings, and angles. Furthermore, we assume the desired framework $f=(\mathcal{G}_N(\mathcal{V}_N,\mathcal{E}_E),q)$ is completely infinitesimally rigid.  
\begin{definition}
Consider a group of $N$ agents in $\mathbb{R}^2$ with an underlying topology described by a graph $\mathcal{G}_N$. The constraints of the desired formation are defined by combinations of displacements, distances, bearings, and angles. A framework $f=(\mathcal{G}_N(\mathcal{V}_N,\mathcal{E}_E),q)$is completely infinitesimally rigid if all its continuous infinitesimally motions are trivial in terms of combinatorial constraints. 
\end{definition}

\begin{figure}
\begin{center}
\label{heterogeneous_diagram}
\includegraphics[width=3.0cm]{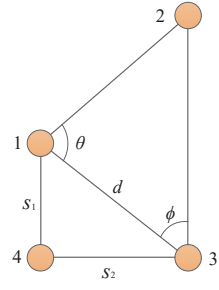}  
\caption{A completely infinitesimally rigid framework with angle constraints $\theta, \phi$, bearing constraint $d$, and distance constraints $s_1, s_2$.} 
\end{center}
\end{figure}

To stabilize all formation constraint variables so that they converge to the desired value, we first give the definition of errors as $$e^h:=[{e^z}^T {e^d}^T {e^b}^T {e^a}^T]^T,$$ where the elements of $e^z, e^d, e^b, e^a$ are defined as same as the foregoing sections. $e^z$ is related to all the displacement constraints, $e^d$ is related to all the distance constraints, $e^b$ is related to all the bearing constraints, and $e^a$ is related to all the angle constraints.
The Hamiltonian for formation stabilization is given as $$H^h=\frac{1}{2}\sum_{\mathcal{S}_z}{e^z}^Te^z+\frac{1}{4}\sum_{\mathcal{S}_d}{e^d}^2+\frac{1}{2}\sum_{\mathcal{S}_b}{e^b}^Te^b+\frac{1}{2}\sum_{\mathcal{S}_a}{e^a}^2, $$
where $\sum_{\mathcal{S}_z}$ is the set of displacement constraints, $\sum_{\mathcal{S}_d}$ is the set of distance constraints, $\sum_{\mathcal{S}_b}$ is the set of bearing constraints, and $\sum_{\mathcal{S}_a}$ is the set of angle constraints.

For the dynamics of the controller, it is proposed by assigning the virtual spring and damping couplings along each constraint to eliminate errors. Establishing the negative feedback connection of the controller system and original system and interconnecting agents by using the incidence matrices and mapping Jacobians, the resulting control law in compact form is derived as
\begin{equation} \label{clh}
\begin{split}
u^h&=-(B^z\otimes I_2)(e^z+D^z\dot{e}_z)-(B^d\otimes I_2){\Omega^d}^T(e^d+D^d\dot{e}^d)\\& \quad -(B^b\otimes I_2){\Omega^b}^T(e^b+D^b\dot{e}^b)-L^a(e^a+D^a\dot{e}^a),
\end{split}
\end{equation}
where $B^z$ consists of the columns of $B$ corresponding to the edges on which there exist the displacement constraints, $B^d$ consists of the columns of $B$ corresponding to the edges on which there exist the distance constraints, $B^b$ consists of the columns of $B$ corresponding to the edges on which there exist the bearing constraints. $L^a$ consists of the rows of $R^a$ corresponding to the angles with constraints. 

The total Hamiltonian with velocity tracking included is given as $H^H=H^v+H^h$. Furthermore, the closed-loop system is given by
\\
\begin{strip}
\begin{align}	\label{clshm}
\begin{split}
\matr{\dot{\bar{q}}  \\ \quad \\ \dot{\bar{p}} \\ \quad \\ \dot{e}^z \\ \quad \\ \dot{e}^d \\ \quad \\ \dot{e}^b \\ \quad \\ \dot{e}^a} =& \matr{0 & I_{2N} & 0 & 0 & 0 & 0\\ \quad \\-I_{2N} & -{D} & -(B^z\otimes I_2) & -(B^d\otimes I_2){\Omega^d}^T & -(B^b\otimes I_2){\Omega^b}^T & -{L^a}^T\\ \quad \\ 0 &{B^z}^T\otimes I_2  & 0 & 0 & 0 & 0\\ \quad \\ 0 & \Omega^d({B^d}^T\otimes I_2)  & 0 & 0 & 0 & 0\\ \quad \\ 0 & \Omega^b({B^b}^T\otimes I_2)  & 0 & 0 & 0 & 0\\ \quad \\ 0 & {L^a}  & 0 & 0 & 0 & 0 }  \matr{ \frac{\partial{H^H}}{\partial{\bar{q}}}(\bar{q},\bar{p}) \\ \quad \\  \frac{\partial{H^H}}{\partial{\bar{p}}}(\bar{q},\bar{p}) \\ \quad \\ \frac{\partial{H^H}}{\partial{{e}^z}}(\bar{q},\bar{p}) \\ \quad \\ \frac{\partial{H^H}}{\partial{e}^d}(\bar{q},\bar{p}) \\\quad \\ \frac{\partial{H^H}}{\partial{e}^b}(\bar{q},\bar{p})\\\quad \\ \frac{\partial{H^H}}{\partial{e}^a}(\bar{q},\bar{p}) },
\end{split}
\end{align} 
\end{strip}

where $D=D^r+D^t+(B^z\otimes I_2)D^f({B^z}^T\otimes I_2)+(B^2\otimes I_2){\Omega^d}^TD^d$ $\Omega^d ({B^d}^T \otimes I_2)+(B^b\otimes I_2){\Omega^b}^TD^b\Omega^b({B^b}^T\otimes I_2)+{L^a}^TD^aL^a.$

To obtain the theorem, we first give the following assumption.
\begin{assumption}
Considering the desired completely infinitesimally rigid framework $f=(\mathcal{G}^t_N(\mathcal{V}_N,\mathcal{E}_E),q^*)$ with a triangulated graph $\mathcal{G}^t_N$, the following matrix associated with heterogeneous constraints of the framework, $$[{B^z}^T\otimes I_2 \quad (B^d\otimes I_2){\Omega^d}^T \quad (B^b\otimes I_2){\Omega^b}^T \quad {L^a}^T],$$ has full column  rank.
\end{assumption}

\begin{remark}
We note that Assumption 1 is not restrictive, as any desired formation can be designed by a completely infinitesimally rigid framework with a triangulated graph. For a single triangular graph with different constraints, it is not difficult to prove that Assumption 1 holds. For the framework with general triangulated graphs, we give an example in Fig. 2 satisfying Assumption 1.
\end{remark}

We obtain the following theorem on rigid formation stabilization with heterogeneous constraints and velocity tracking.

\begin{theorem} \label{theo5}
Consider a group of agents modeled in port-Hamiltonian form as described by \eqref{dyph}. If the desired framework $f=(\mathcal{G}^t_N(\mathcal{V}_N,\mathcal{E}_E),q^*)$ with a triangulated graph $\mathcal{G}^t_N$ is completely infinitesimally rigid, then using the control law $u^v$ (\ref{uv})+$u^h$ (\ref{clh}), the group of the agents converges to the desired formation defined by different constraints locally and asymptotically, and each agent tracks the desired velocity. 
\end{theorem}

\begin{proof}
Take the Hamiltonian $H^H$ as a Lyapunov candidate. It follows that $H^H \ge 0$. Taking the time derivative of the Hamiltonian, we have
$$\dot{H}^H=-\frac{\partial^T H^H}{\partial \bar{p}}D \frac{\partial H^H}{\partial \bar{p}}.$$ 
It follows that $\dot{H}^H \leq 0$. Furthermore, invoking LaSalle Invariance Principle gives that the system converges to the largest invariant set where $\bar{p}=0$. It follows $\dot{\bar{p}}=0$. Substituting $\bar{p}=0, \dot{\bar{p}}=0$ into (\ref{clshm}) gives
\begin{align} \label{doth}
\begin{split}
-({B^z}^T\otimes I_2)e^z - (B^d\otimes I_2){\Omega^d}^Te^d \\ - (B^b\otimes I_2){\Omega^b}^Te^b-{R^a}^Te^a=0
\end{split}
\end{align}
By invoking Assumption 1, the errors $$e^h=[{e^z}^T {e^d}^T {e^b}^T {e^a}^T]^T$$ converge to 0, thus completing
the proof.
\end{proof}

\section{Simulation Examples}
In order to illustrate the proposed framework, we consider a group of four agents modeled as in \eqref{dyph}, which is a fully actuated system. The corresponding parameters in the model dynamics are given as follows: $m_i=1, i=1,2,3,4.$ $D^r_i={\rm diag}(0,0), i=1,2,3,4.$ The simulation conditions are set as follows.

\textit{(1) Displacement-based formation.} We consider four agents  interconnected by a line (acyclic) graph and a (cyclic) ring graph, respectively. The network diagrams of displacement-based formations are shown in Fig. 3. Accordingly, the incidence matrices $B_1$ and $B_2$ take the following forms 
$$B_1=\matr{-1 & 0 & 0 & 1\\1 & -1 & 0 & 0 \\ 0 & 1 & -1 & 0 \\ 0 & 0 & 1 & -1}, B_2=\matr{-1 & 0 & 0 \\1 & -1 & 0 \\ 0 & 1 & -1 \\ 0 & 0 & 1 }.$$

The initial positions and velocities of the four agents are given by $q_1(0)=(1,1)$, $q_2(0)=(2,1)$, $q_3(0)=(2,2)$, $q_4(0)=(3,2)$ and $v_i(0)=(0,0), i=1,2,3,4.$ The desired displacements of the formations with acyclic graph and cyclic graph are given by $z_1^*=(-1,1),z_2^*=(1,1),z_3^*(1,-1)$ and $z_1^*=(-1,1),z_2^*=(1,1),z_3^*(1,-1),z_4^*=(-1,-1)$.  The desired velocities are given by $v_i^*=(1,1), i=1,2,3,4.$ The corresponding parameters in the controllers are given as follows: $D^t_i={\rm diag}(1,1), i=1,2,3,4.$ $D_j^f={\rm diag}(1,1), j=1,2,3,4$ for acyclic graph and $j=1,2,3,4,5$ for cyclic graph.

By applying the control laws (\ref{uv}) and (\ref{uc}), Fig. 4 is obtained, which shows the evolution of displacement-based formation errors and velocity tracking errors with acyclic and cyclic graphs. In subfigures (a) and (b), $z_{jx}, z_{jy}, j=1,2,3,4,5$ denote the corresponding displacement errors of edge $j$ along the $x$ and $y$ axis, respectively. In subfigures (c) and (d), $v_{ix}, v_{iy}, i=1,2,3,4$ denote the velocity tracking errors of each agent $i$ along the $x$ and $y$ axis, respectively. It can be observed that all the errors converge to zero, which illustrates Theorem \ref{theo1}. 

\begin{figure}[htbp]    \label{tpdis}
    \centering
\subfigure[Displacement-based formation with acyclic graph]{\includegraphics[width=4cm]{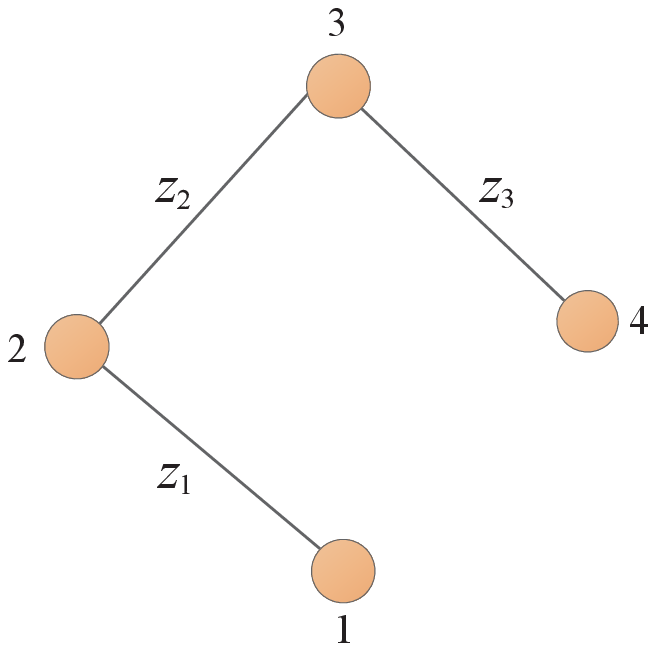}}
\subfigure[Displacement-based formation with cyclic graph]{\includegraphics[width=4cm]{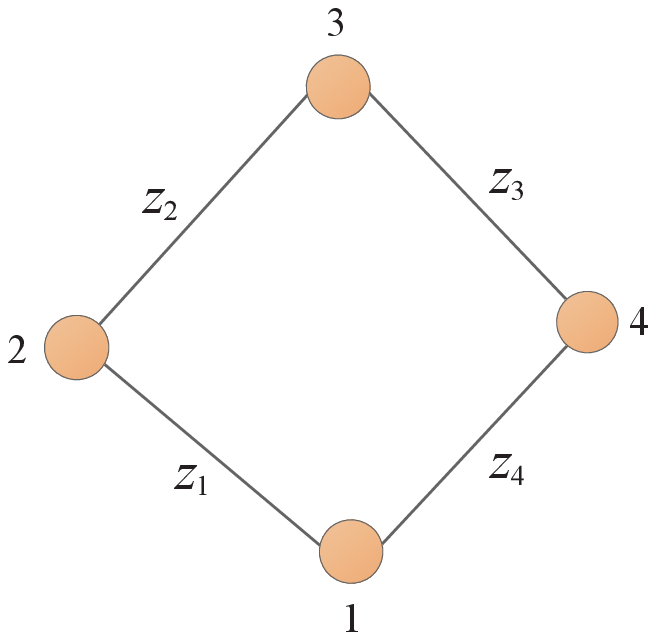}}
    \quad 
    \caption{Network diagrams of displacement-based formations}
\end{figure}

\begin{figure}[htbp]    \label{disev}
    \centering
    \subfigure[Formation errors with acyclic graph]{\includegraphics[width=4cm]{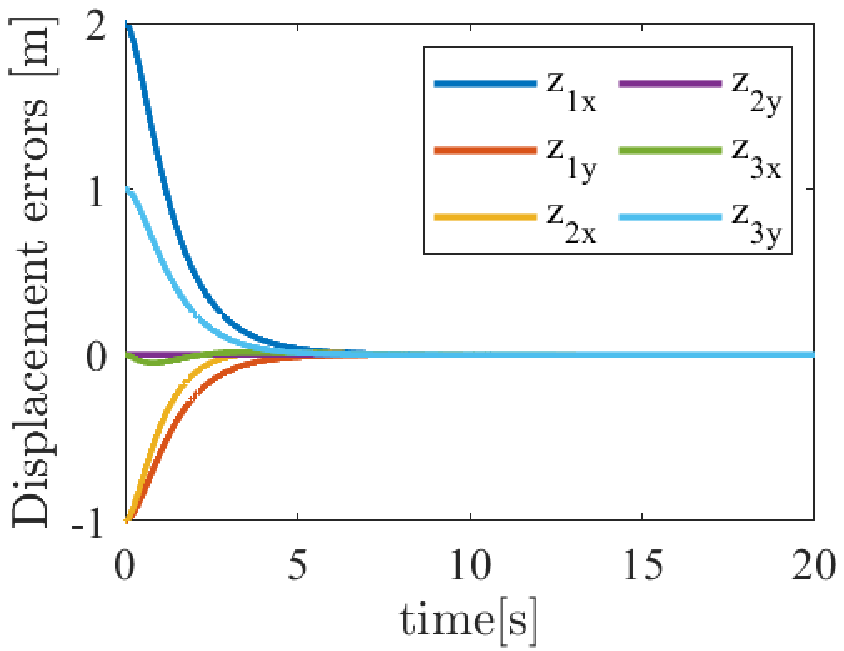}}
    \subfigure[Formation errors with cyclic graph]{\includegraphics[width=4cm]{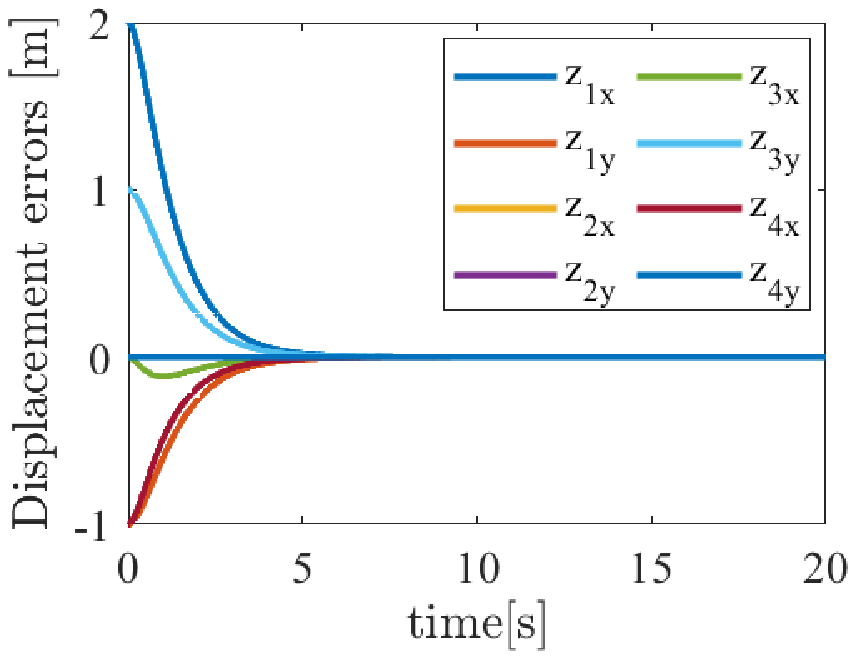}}
    \quad 
    \subfigure[Velocity errors with acyclic graph]{\includegraphics[width=4cm]{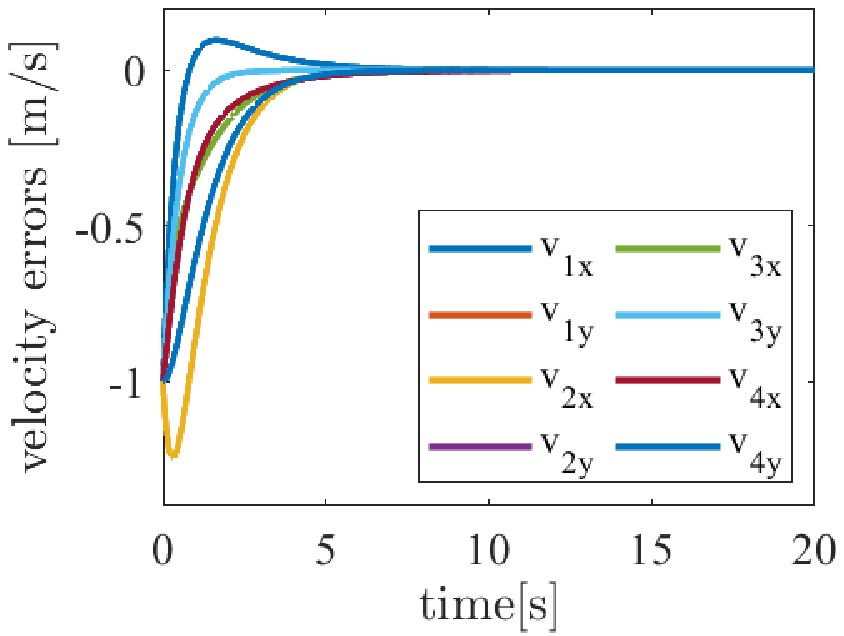}}
    \subfigure[Velocity errors with cyclic graph]{\includegraphics[width=4cm]{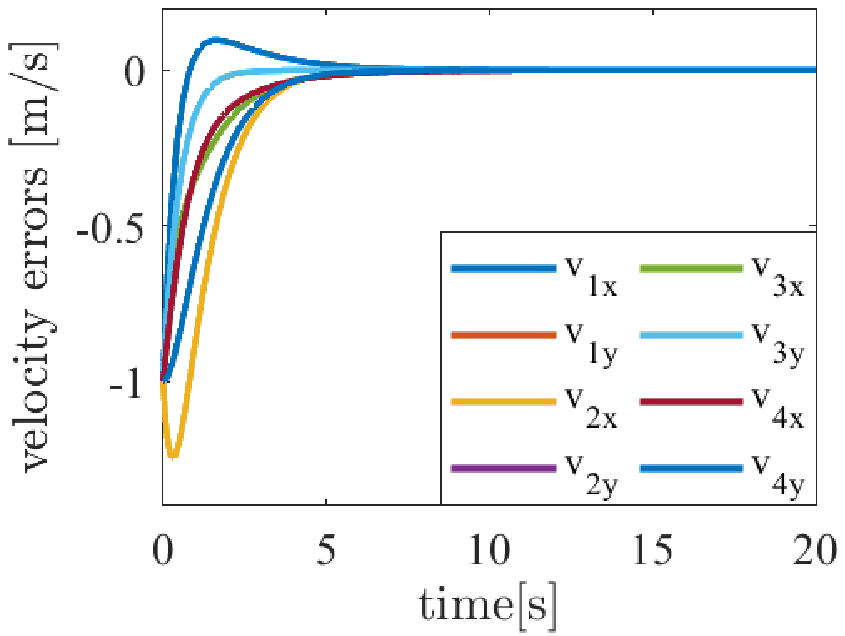}}
    \quad 
    \caption{Evolution of the displacement-based formations}
\end{figure}

\textit{(2) Distance-based formation.} We assume the framework is infinitesimally distance rigid. The network diagram of distance-based formation is shown in Fig. 5 and the corresponding incidence matrix takes the form as
$$B_3=\matr{-1 & 0 & 0 & 1 & 0 \\1 & -1 & 0 & 0 & 1\\ 0 & 1 & -1 & 0 & 0\\ 0 & 0 & 1 & -1 & -1}.$$

The initial positions and velocities of the four agents are given by $q_1(0)=(1,1),q_2(0)=(2,1),q_3(0)=(2,2),q_4(0)=(3,2)$ and $v_i(0)=(0,0), i=1,2,3,4.$ The desired distances of the formation are given by $d_1^*=\frac{\sqrt{2}}{2},d_2^*=\frac{\sqrt{2}}{2},d_3^*=\frac{\sqrt{2}}{2},d_4^*=\frac{\sqrt{2}}{2},d_5^*=2$. The desired velocities are given as $v_i^*=(1,1), i=1,2,3,4.$ The corresponding parameters in the controllers are given as follows: $D_j^d=1, j=1,2,3,4,5.$

By applying the control laws (\ref{uv}) and (\ref{cld}), Fig. 6 is obtained which shows the evolution of distance-based formation errors and velocity tracking errors. In subfigures (a) and (b), $d_{j}, j=1,2,3,4,5$ denote the corresponding distance errors of edge $j$ and $v_{ix}, v_{iy}, i=1,2,3,4$ denote the velocity tracking errors of each agent $i$ along the $x$ and $y$ axis, respectively. It can be observed that all the errors converge to zero, which verifies Theorem \ref{theo2}. 

\begin{figure}[htbp]    \label{tpd}
\begin{center}
\includegraphics[width=4cm]{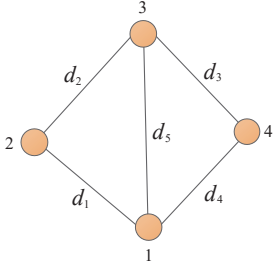}  
\caption{Network diagram of distance-based formation} 
\end{center}
\end{figure}

\begin{figure}[htbp]    \label{dev}
    \centering
    \subfigure[Formation errors]{\includegraphics[width=4cm]{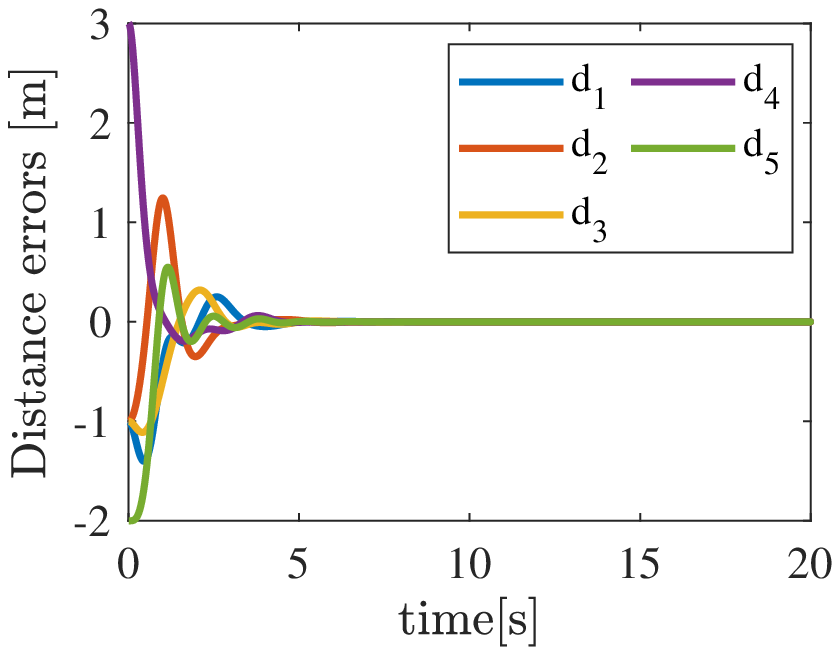}}
    \subfigure[Velocity tracking errors]{\includegraphics[width=4cm]{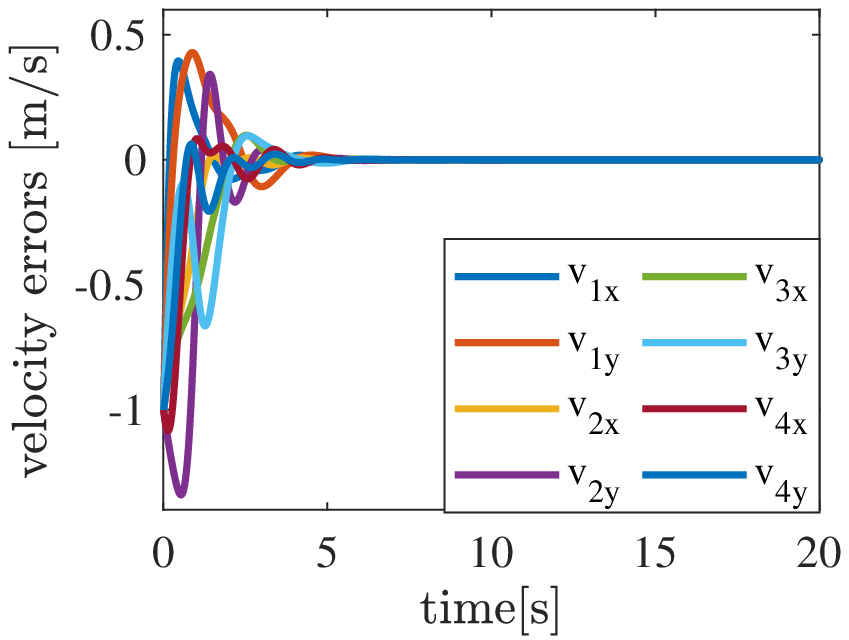}}
    \quad 
    \caption{Evolution of the distance-based formation}
\end{figure}

\textit{(3) Bearing-based formation.} We assume the framework is infinitesimally bearing rigid. The network diagram of bearing-based formation is shown in Fig. 7 and the corresponding incidence matrix takes the form as
$$B_3=\matr{-1 & 0 & 0 & 1 & 0 \\1 & -1 & 0 & 0 & 1\\ 0 & 1 & -1 & 0 & 0\\ 0 & 0 & 1 & -1 & -1}.$$

The initial positions and velocities of the four agents are given by $q_1(0)=(1,1),q_2(0)=(2,1),q_3(0)=(2,2),q_4(0)=(3,2)$ and $v_i(0)=(0,0), i=1,2,3,4.$ The desired bearings of the formation are given by $s_1^*=(-\frac{\sqrt{2}}{2},\frac{\sqrt{2}}{2}),s_2^*=(\frac{\sqrt{2}}{2},\frac{\sqrt{2}}{2}),s_3^*=(\frac{\sqrt{2}}{2},-\frac{\sqrt{2}}{2}), s_4^*=(-\frac{\sqrt{2}}{2},-\frac{\sqrt{2}}{2}),s_5^*=(-1,0)$. The desired velocities are given as $v_i^*=(1,1), i=1,2,3,4.$ The corresponding parameters in the controllers are given as follows: $D_j^b={\rm diag}(1,1), j=1,2,3,4,5.$

By applying the control laws (\ref{uv}) and (\ref{clb}), Fig. 8 is obtained which shows the evolution of bearing-based formation errors and velocity tracking errors. In subfigures (a) and (b),  $s_{jx}, s_{jy}, j=1,2,3,4,5$ denote the corresponding bearing errors of edge $j$ and $v_{ix}, v_{iy}, i=1,2,3,4$ denote the velocity tracking errors of each agent $i$ along the $x$ and $y$ axis, respectively. It can be observed that all the errors converge to zero, which illustrates Theorem \ref{theo3}. 

\begin{figure}[ht]    \label{tpbearing}
\begin{center}
\includegraphics[width=4cm]{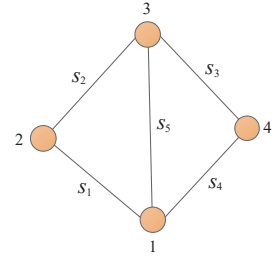}  
\caption{Network diagram of bearing-based formation} 
\end{center}
\end{figure}

\begin{figure}[ht]    \label{bev}
    \centering
    \subfigure[Formation errors]{\includegraphics[width=4cm]{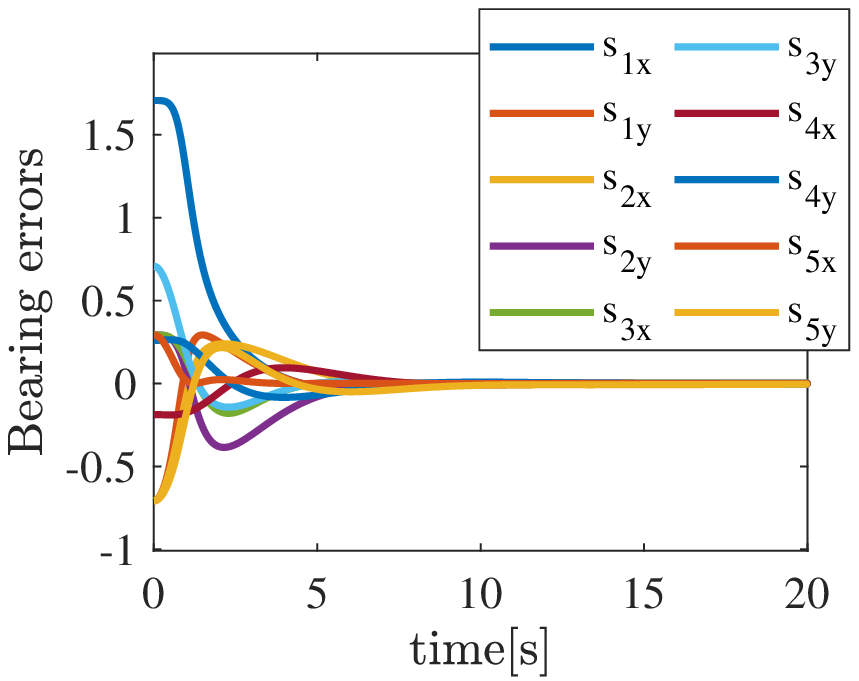}}
    \subfigure[Velocity tracking errors]{\includegraphics[width=4cm]{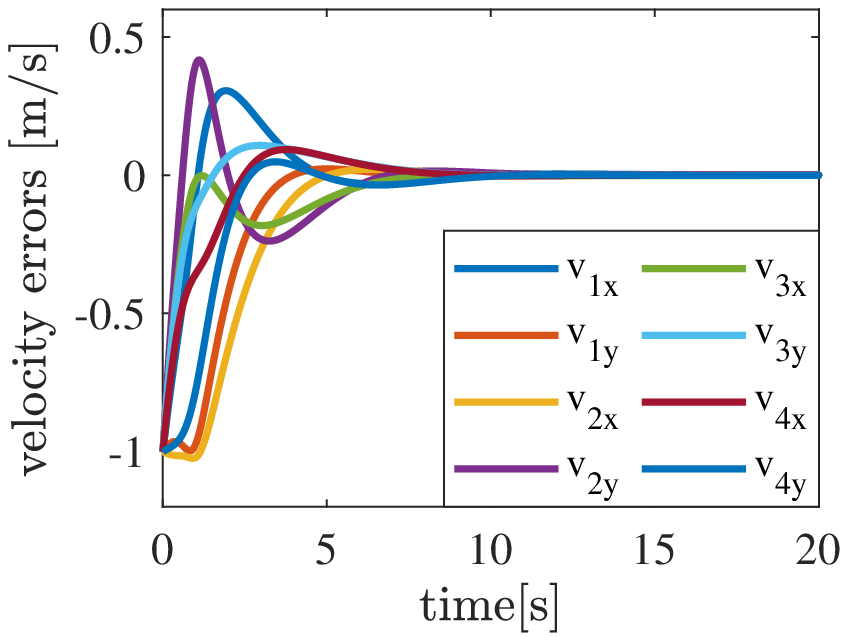}}
    \quad 
    \caption{Evolution of the bearing-based formation}
\end{figure}

\textit{(4) Angle-based formation.} We consider a triangulated Laman graph and assume the framework is infinitesimally angle rigid. The network diagram of angle-based formation is shown in Fig. 9.

The initial positions and velocities of the four agents are given by $q_1(0)=(3.1357,3.1311)$, $q_2(0)=(4.1515,4.0342)$, $q_3(0)=(4.1486,2.1412)$, $q_4(0)=(3.0784,0.0064)$ and $v_i(0)=(0,0), i=1,2,3,4.$ The desired angles of the formation are given by $\theta_1=\frac{\pi}{2},\phi_1^*=\frac{\pi}{4},\theta _2^*=\frac{\pi}{4}, \phi_2^*=\frac{\pi}{4}$. The desired velocities are given by $v_i^*=(3,3), i=1,2,3,4.$ The corresponding parameters in the controllers are given as follows: $d_{l\theta}^a=3,d_{l\phi}^a=3, l=1,2.$

By applying the control laws (\ref{uv}) and (\ref{cla}), Fig. 10 is obtained which shows the evolution of angle-based formation errors and velocity tracking errors. In subfigures (a) and (b),  $\theta_{1}, \theta_{2}, \phi_{1}, \phi_{2}$ denote the corresponding angle errors and $v_{ix}, v_{iy}, i=1,2,3,4$ denote the velocity tracking errors of each agent $i$ along the $x$ and $y$ axis, respectively. It can be observed that all the errors converge to zero, which illustrates Theorem \ref{theo4}. 

\begin{figure}[ht]    \label{tpangle}
\begin{center}
\includegraphics[width=4cm]{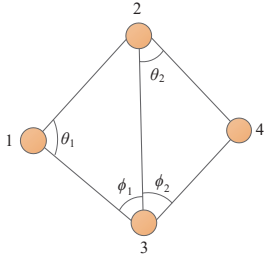}  
\caption{Network diagram of Angle-based formation} 
\end{center}
\end{figure}

\begin{figure}[ht]    \label{aev}
    \centering
    \subfigure[Formation errors]{\includegraphics[width=4cm]{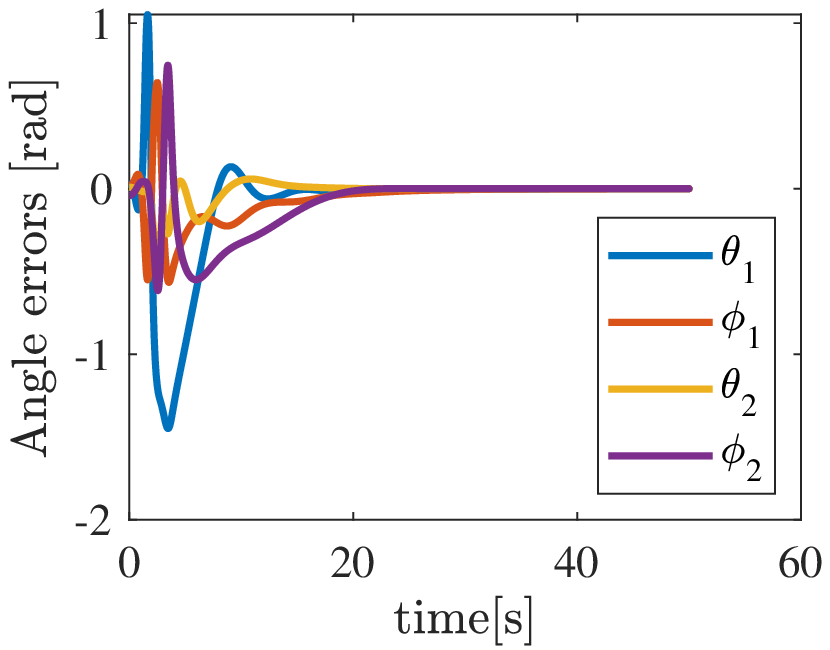}}
    \subfigure[Velocity tracking errors]{\includegraphics[width=4cm]{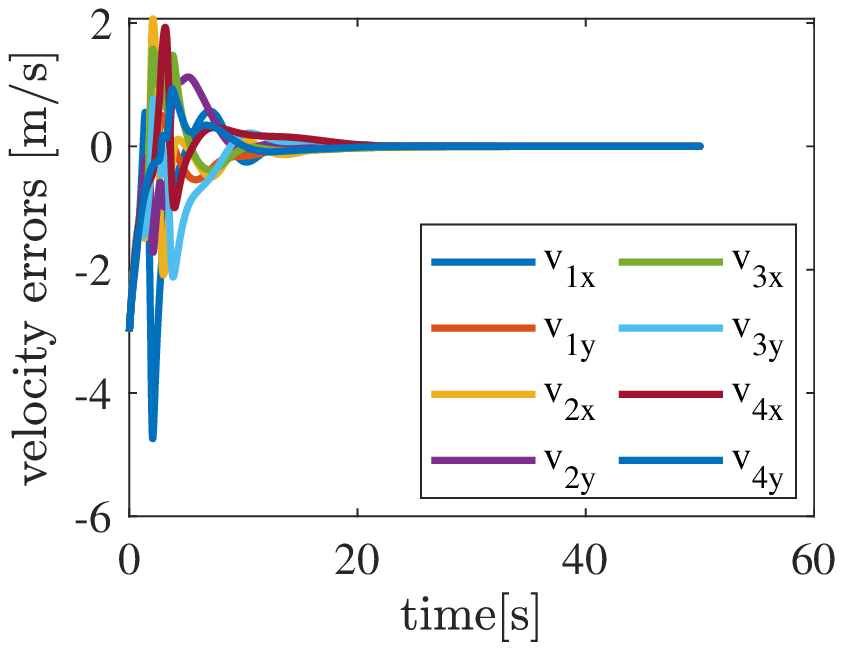}}
    \quad 
    \caption{Evolution of the angle-based formation}
\end{figure}

\begin{figure}[!ht]    \label{hev}
    \centering
    \subfigure[Formation errors]{\includegraphics[width=4cm]{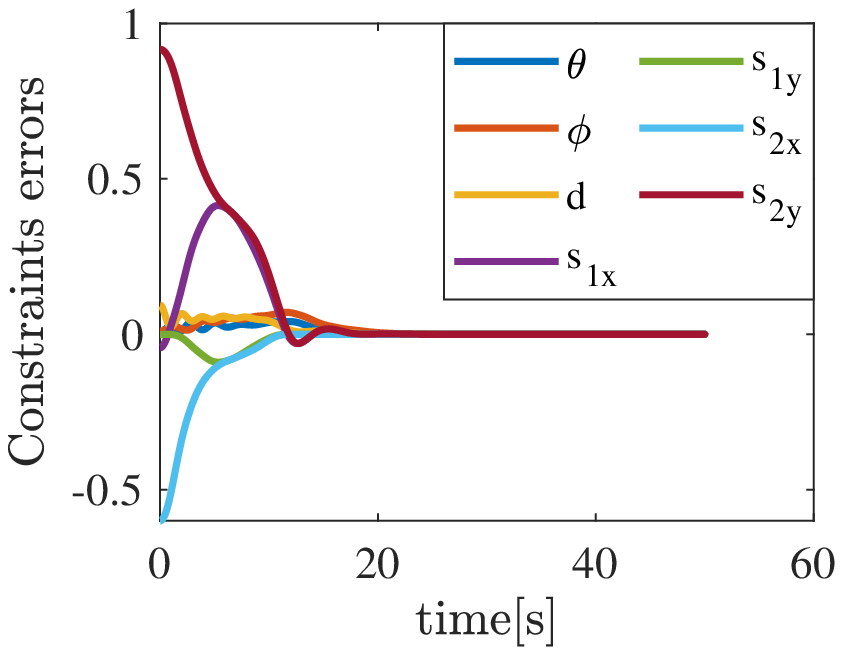}}
    \subfigure[Velocity tracking errors]{\includegraphics[width=4cm]{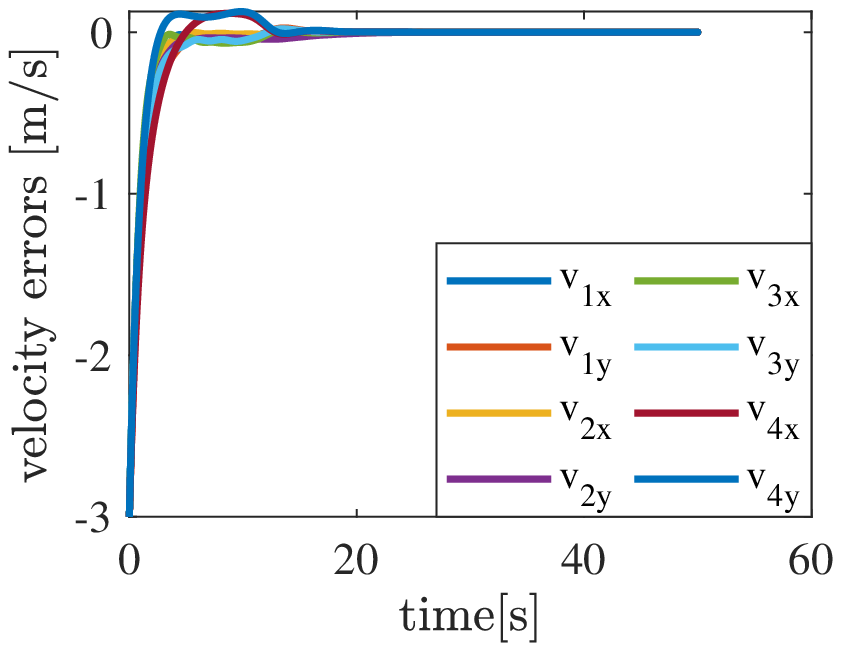}}
    \quad 
    \caption{Evolution of the formation with heterogeneous constraints}
\end{figure}

\textit{(5) Formation with heterogeneous constraints.} We consider a triangulated Laman graph and assume the framework is completely infinitesimally rigid. The network diagram of the formation with heterogeneous constraints is shown in Fig. 2.

The initial positions and velocities of the four agents are given by $q_1(0)=(1.0844,2.1311)$, $q_2(0)=(2.1831,3.0071)$, $q_3(0)=(2.1584,1.1698)$, $q_4(0)=(3.1919,2.1868)$ and $v_i(0)=(0,0), i=1,2,3,4.$ The desired constraints of the formation are given by $\theta^*=\frac{\pi}{2},\phi^*=\frac{\pi}{4}, d^*=\sqrt{2}, s_1^*=(0,1),s_2^*=(1,0)$. The desired velocities are given by $v_i^*=(3,3), i=1,2,3,4.$ The corresponding parameters in the controllers are given as follows: $d_{\theta}^a=3,d_{\phi}^a=3, D^d=1, D_1^b={\rm diag}(1,1), D_2^d={\rm diag}(1,1).$

By applying the control laws (\ref{uv}) and (\ref{clh}), Fig. 11 is obtained which shows the evolution of formation errors and velocity tracking errors. In subfigures (a),  $\theta, \phi, \d, s_1$, and $s_{2}$ denote the corresponding angle, distance, and bearing errors and $v_{ix}, v_{iy}, i=1,2,3,4$ denote the velocity tracking errors of each agent $i$ along the $x$ and $y$ axis, respectively. It can be observed that all the errors converge to zero, which illustrates Theorem \ref{theo5}. 

\section{Conclusions}
In this paper, a passivity approach in pH form for formation stabilization and velocity tracking is proposed. For displacement-based formation, the proposed framework can be applied to not only   acyclic graphs but also   cyclic graphs. For rigid formations defined by one kind of constraint, the local convergence of the proposed formation system is guaranteed by establishing the relationship between infinitesimal rigidity and the time derivative of the Hamiltonian. For rigid formations defined by a combination of heterogeneous  constraints, the local convergence to a desired formation shape is ensured by infinitesimally completely rigid conditions, which should be further investigated for future research. Several simulations are given to illustrate the validity of the framework.   






\bibliography{root}   
\bibliographystyle{unsrt}

\end{document}